\documentclass[10pt, a4paper, oneside]{article}
\usepackage[top=1.3cm, left=3.2cm, right=3.2cm]{geometry}
\usepackage{etoolbox}

\usepackage{booktabs}
        
\usepackage{microtype}
\usepackage{tikz,pgfplots}
\pgfplotsset{compat=1.10}
\usepgfplotslibrary{fillbetween}
\usetikzlibrary{positioning}

\usepackage{algorithm,algorithmicx,amssymb,amsmath,amsthm,amsfonts,bm,bbm,epsfig,dsfont,authblk}
\usepackage[labelfont=bf,font=normalsize,indention=0.6cm,margin=0cm]{caption}
\usepackage[noend]{algpseudocode}
\usepackage[numbers,sort&compress,longnamesfirst,sectionbib]{natbib}
\usepackage{multirow,xspace,tabularx,xcolor,nicefrac,lscape,hyperref}

\usepackage{thmtools,thm-restate}
\newtheorem{theorem}{Theorem}
\newtheorem{lemma}{Lemma}
\newtheorem{definition}{Definition}
\newtheorem{corollary}{Corollary}

\theoremstyle{plain}

\numberwithin{theorem}{section}
\numberwithin{lemma}{section}
\numberwithin{definition}{section}
\numberwithin{corollary}{section}
\numberwithin{equation}{section}

\newcommand{\eps}{\varepsilon}
\newcommand{\argmin}{\operatorname{argmin}}

\renewcommand{\leq}{\leqslant}
\renewcommand{\geq}{\geqslant}
\renewcommand{\le}{\leqslant}
\renewcommand{\ge}{\geqslant}

\newcommand{\Oh}{\mathcal{O}}

\renewcommand{\algref}[1]{Algorithm~\ref{alg:#1}}
\newcommand{\thmref}[1]{Theorem~\ref{thm:#1}}
\newcommand{\thmrefs}[2]{Theorems~\ref{thm:#1} and~\ref{thm:#2}}
\newcommand{\lemref}[1]{Lemma~\ref{lem:#1}}
\newcommand{\lemrefs}[2]{Lemmas~\ref{lem:#1} and~\ref{lem:#2}}
\newcommand{\corref}[1]{Corollary~\ref{cor:#1}}
\newcommand{\defref}[1]{Definition~\ref{def:#1}}
\newcommand{\figref}[1]{Figure~\ref{fig:#1}}
\newcommand{\tabref}[1]{Table~\ref{tab:#1}}
\newcommand{\secref}[1]{Section~\ref{sec:#1}}
\newcommand{\eq}[1]{Equation~\eqref{eq:#1}}
\newcommand{\eqs}[2]{Equations~\eqref{eq:#1} and~\eqref{eq:#2}}

\newcommand{\Pro}[1]{\mathrm{Pr} [ #1 ]}
\newcommand{\Ex}[1]{\mathbb{E}\left[\,#1\,\right]}
\newcommand{\dif}{\,\mathrm{d}}

\newcommand{\betabound}{\beta \geq \tfrac{2k-1}{k-1} + \eps}
\newcommand{\Wlog}{W.\,l.\,o.\,g.\xspace}
\newcommand{\wlo}{w.\,l.\,o.\,g.\xspace}
\newcommand{\ie}{i.\,e.\xspace}
\newcommand{\eg}{e.\,g.\xspace}
\newcommand{\whp}{w.\thinspace h.\thinspace p.\xspace}

\title{Bounds on the Satisfiability Threshold for Power~Law Distributed Random SAT}

\author[1]{Tobias Friedrich}
\author[1]{Anton Krohmer}
\author[1]{Ralf Rothenberger}
\author[2]{Thomas Sauerwald}
\author[1]{Andrew M. Sutton}
\affil[1]{\small Hasso Plattner Institute, Potsdam, Germany\\
  \texttt{firstname.lastname@hpi.de}}
\affil[2]{\small University of Cambridge, Cambridge, United Kingdom\\
  \texttt{firstname.lastname@cl.cam.ac.uk}}
\date{}

\begin{document}
\clearpage\maketitle
\thispagestyle{empty}

\begin{abstract}
Propositional satisfiability (SAT) is one of the most fundamental problems in computer science.
The worst-case hardness of SAT
lies at the core of 
computational complexity theory.
The average-case analysis of SAT has triggered the development of
sophisticated rigorous and non-rigorous techniques for analyzing
random structures.

Despite a long line of research and substantial progress, nearly all theoretical work
on random SAT assumes a \emph{uniform} distribution on the variables.
In contrast, real-world instances often exhibit large fluctuations in variable occurrence.
This can be modeled by a \emph{scale-free} distribution of the variables,
which results in distributions closer to industrial SAT instances.

We study random $k$-SAT on $n$~variables, $m=\Theta(n)$~clauses, and
a power law distribution on the variable occurrences with exponent $\beta$.
We observe a satisfiability threshold at $\beta=(2k-1)/(k-1)$.
This threshold is tight in the sense that
instances with $\beta\leq(2k-1)/(k-1)-\varepsilon$
for any constant $\eps>0$ 
are \emph{unsatisfiable}
with high probability (\whp).
For $\beta\geq(2k-1)/(k-1)+\varepsilon$,
the picture is reminiscent of the uniform case:
instances are \emph{satisfiable} \whp\
for sufficiently small constant
clause-variable ratios $m/n$;
they are \emph{unsatisfiable}
above a ratio $m/n$
that depends on $\beta$.
\end{abstract}
\newpage

\section{Introduction}

Satisfiability of propositional formulas (SAT) is one of the most researched problems in theoretical computer science.
SAT is widely used to model practical problems such as bounded model checking, hardware and software verification, automated planning and scheduling, and circuit design.
Even \emph{large industrial instances} with millions of variables can often be
solved very efficiently by modern SAT solvers.
The structure of these industrial SAT instances appears to allow a much faster processing than
the theoretical worst-case of this NP-complete problem.
It is an open and widely discussed question which structural properties make a SAT instance easy to solve
for modern SAT solvers.

\medskip\noindent
{\bf Random SAT:}
For modeling typical inputs, we study random propositional formulas.
In random satisfiability, we have a distribution over Boolean formulas in conjunctive normal form (CNF).
The degree of a variable in a CNF formula is the number of disjunctive clauses in which that variable appears either positively or negatively.
Two interesting properties of random models are its \emph{degree distribution} and its \emph{satisfiability threshold}.
The degree distribution $F(x)$ of a formula~$\Phi$ is the fraction of variables that occur more than $x$ times (negated or unnegated).
A satisfiability threshold is a critical value around which the probability that a formula is satisfiable changes from~$0$ to~$1$.

\medskip\noindent
{\bf Uniform random SAT:}
In the classical uniform random model, the degree distribution is binomial.
On uniform random $k$-SAT, the \emph{satisfiability threshold conjecture}~\cite{Achlioptas2011geometry}
asserts if $\Phi$ is a formula drawn uniformly at random from the set of
all $k$-CNF formulas with $n$ variables and $m$ clauses, there exists
a real number $r_k$ such that
\[
\lim_{n \to \infty} \Pr\{\Phi \text{~is satisfiable}\} =
\begin{cases}
  1 & m/n < r_k;\\
  0 & m/n > r_k.
\end{cases}
\]
A well-known result of Friedgut~\cite{Friedgut1999thresholds} establishes that the transition is sharp, even though its location is not known exactly for all values of $k$ (and may also depend on $n$). For $k=2$, the critical threshold is $r_2 = 1$~\cite{chvatalreed92,FernandezdelaVega2001random2sat,Goerdt1996threshold}.
Recently, Coja-Oghlan and Panagiotou \cite{Coja-Oghlan:2014:AKT:2591796.2591822,kostathreshold} gave a sharp bound (up to lower order terms) with 
$r_k = 2^k \log 2 - \tfrac12 (1 + \log 2) \pm o_k(1)$.
Ding, Sly, and Sun~\cite{Ding:2015:PSC:2746539.2746619} derive an exact representation of the threshold for all $k\geq k_0$, where $k_0$ is a large enough constant.
Explicit bounds also exist for low values of $k$, e.g., $3.52 \leq r_3 \leq 
4.4898$~\cite{HajiaghayiSorkin2003threshold,Kaporis2006probabilistic,DBLP:journals/tcs/DiazKMP09},
and numerical estimates using the cavity method from statistical mechanics~\cite{mezard2002analytic} suggest that $r_3 \approx 4.26$.

\medskip\noindent
{\bf Other random SAT models:} In the regular random model~\cite{BDIS2006}, formulas are constructed at random, but the degree distribution is fixed: each literal appears exactly $\lfloor \frac{km}{2n} \rfloor$ or $\lfloor \frac{km}{2n} \rfloor + 1$ times in the formula. Similarly, Bradonjic and Perkins~\cite{bradonjic_et_al:LIPIcs:2014:4719} considered a random geometric $k$-SAT model in which $2n$ points are placed at random in $[0,1]^d$. Each point corresponds to a unique literal, and clauses are formed by all $k$-sets of literals that lie together within a ball of diameter $\Theta(n^{-1/d})$. Again, this model has a binomial variable distribution.

\medskip\noindent
{\bf Power law random SAT:}
Recently, there has been a paradigm shift when modeling real-world data. In many applications, it has been found that certain quantities do not cluster around a specific scale as suggested by a uniform distribution, but are rather inhomogeneous~\cite{clauset2009power,newman2005power}. In particular, the degree distribution in complex networks often follows a power law~\cite{Newman03}. This means that the fraction of vertices of degree $k$ is proportional to $k^{-\beta}$, where the constant $\beta$ depends on the network. To mathematically study the behavior of such networks, random graph models that generate a power law degree distribution have been proposed~\cite{BarabasiAlbert1999,PhysRevEDmitri,chunglupower,soderberg2002general}. 

While there has been a large amount of research on power law random graphs in the past few years~\cite{Hofstad},
there is little previous work on power law SAT formulas. Nevertheless, the observation that quantities follow a power law in real-world data has also emerged in the context of SAT~\cite{BDIS2006}.
As all aforementioned random SAT models assume strongly concentrated degree distributions, it was conjectured that this property might be modeled well by random formulas with a power law degree distribution.

To address this conjecture, and to help close the gap between the structure of uniform random and industrial instances, Ansótegui, Bonet, and Levy~\cite{ABL2009} recently proposed a power-law random SAT model. This model has been studied experimentally~\cite{ABL2009,AnsoteguiBL09,DBLP:conf/cade/AnsoteguiBGL14,DBLP:conf/ccia/AnsoteguiBGL15}, and empirical investigations found that (1) indeed the constraint graphs of many families of industrial instances obey a power-law and (2) SAT solvers that are constructed to specialize on industrial instances perform better on power-law formulas than on uniform random formulas.
To complement these experimental findings, we contribute with this paper the first theoretical results on this model.

\begin{figure}
  \def\plotwidth{5.8cm}
      \centering      
      \definecolor{darkgreen}{RGB}{0,128,0}
      \definecolor{darkred}{RGB}{128,0,0}
    
      \begin{tikzpicture}
       
        \begin{axis}[
          name=plot1,
          width=\plotwidth,
          scale only axis,
          no marks,
          xlabel={$\beta$},
          ylabel={$m/n$},
          ylabel style={yshift=-2.5mm},
          yticklabels={,,},
          y tick style={draw=none},
          extra y ticks={2.5},
          extra y tick style={ticklabel pos=right},
          xtick={2.5},
          xticklabels={\tiny $\tfrac{2k-1}{k-1}$},
          xmin=2.2,
          xmax=2.75,
          ymin=1,
          ymax=4.2,          
          ]
          
          \plot[line width=1.6pt,color=black,dashed,domain=2.5:4]
          table{singleflip.dat} 
          node[pos=0.2, inner sep=0pt] (singleflip-anchor) {};
          
          \addplot[line width=1.6pt, color=black, dashed] 
          coordinates {(2.5,2.5) (4,2.5)} 
          node[pos=0.086,inner sep=0pt] (sat-anchor-r) {};
          
          \addplot[line width=1.6pt, color=black, dashed] 
          coordinates {(2.5,0) (2.5,3.5)} 
          node[pos=0.8, inner sep=0pt] (unsat-anchor) {}
          node[pos=0.4, inner sep=0pt] (sat-anchor) {};

          \path[name path=left,draw=none] 
          (axis cs:0,0) -- (axis cs:0,\pgfkeysvalueof{/pgfplots/ymax});
          
          \path[name path=bottom,draw=none] 
          (axis cs:0,0) -- (axis cs:\pgfkeysvalueof{/pgfplots/xmax},0);
          
          \path[name path=unsat,draw=none] 
          (axis cs:2.5,0) -- (axis cs:2.5,2.5) -- (axis cs:\pgfkeysvalueof{/pgfplots/xmax},2.5);

          \path[name path=sat,draw=none] 
          (axis cs:2.5,0) -- (axis cs:2.5,2.5) -- (axis cs:\pgfkeysvalueof{/pgfplots/xmax},2.5);

          \addplot[draw=none,name path=singleflip] 
          table{singleflip-extended.dat};

          \addplot[fill=darkred!30] 
          fill between[of=left and singleflip];
          
          \addplot[fill=darkgreen!30] 
          fill between[of=sat and bottom];
          
          \path[name path=join,draw=none] 
          (axis cs:2.5,2.5) -- (axis cs:\pgfkeysvalueof{/pgfplots/xmax},2.5);
          
          \addplot[fill=black!30] fill between[of=singleflip and join];

          \node[below,draw,fill=white] (T1) at (axis cs:2.625,2) 
          {\sffamily \bfseries \hyperref[thm:sat]{Theorem}~\ref{thm:sat}};
          \draw[-latex,thick, bend left] (T1) to (sat-anchor);
          \draw[-latex,thick] (T1) to (sat-anchor-r);
          
          \node[left,draw,fill=white] (T2) at (axis cs:2.45,2) 
          {\sffamily \bfseries \hyperref[thm:unsat]{Theorem}~\ref{thm:unsat}};
          \draw[-latex,thick] (T2) to (unsat-anchor);
          
          \node[left,draw,fill=white] (T3) at (axis cs:2.45,3.4) 
          {\sffamily \bfseries \hyperref[thm:singleflip-powerlaw]{Theorem}~\ref{thm:singleflip-powerlaw}};
          \draw[-latex,thick] (T3) [out=45, in=150] to (singleflip-anchor);
          
          \node[above left, anchor=south east] at 
          (axis cs:\pgfkeysvalueof{/pgfplots/xmax},\pgfkeysvalueof{/pgfplots/ymin}) 
          {\emph{satisfiable}};
          
          \node[below right, anchor=north west] at 
          (axis cs:\pgfkeysvalueof{/pgfplots/xmin},\pgfkeysvalueof{/pgfplots/ymax}) 
          {\emph{unsatisfiable}};

          \node[below left, anchor=north east, yshift=-8mm] at 
          (axis cs:\pgfkeysvalueof{/pgfplots/xmax},\pgfkeysvalueof{/pgfplots/ymax}) 
          {\emph{unknown}};

      \end{axis}
      \begin{axis}[
        at=(plot1.right of south east),
        anchor=left of south west,
        xshift=10mm,
          width=\plotwidth,
          scale only axis,
          xlabel={$\beta$},
          extra x ticks={2.5},
          extra x tick labels={\vphantom{\tiny $\frac{2k-1}{k-1}$}},
          extra x tick style={draw=none},
          ylabel={$m/n$},
          ylabel style={yshift=-2mm},
          xmin=2,
          xmax=3,
          ymin=1,
          ymax=5,
          ]

            \addplot[line width=1.6, black, dashed, name path=US]
            table{unsat-sat-boundary.dat};
            \addplot[line width=1.6, black, dashed, name path=HS]
            table{hard-sat-boundary.dat}{};
            \addplot[line width=1.6, black, dashed, name path=HU]
            table{hard-unsat-boundary.dat};

          \node[above left, anchor=south east] at 
          (axis cs:\pgfkeysvalueof{/pgfplots/xmax},\pgfkeysvalueof{/pgfplots/ymin}) 
          {\emph{satisfiable}};
          
          \node[below right, anchor=north west] at 
          (axis cs:\pgfkeysvalueof{/pgfplots/xmin},\pgfkeysvalueof{/pgfplots/ymax}) 
          {\emph{unsatisfiable}};

          \node[below left, anchor=north east, yshift=-8mm] (unknown) at 
          (axis cs:\pgfkeysvalueof{/pgfplots/xmax},\pgfkeysvalueof{/pgfplots/ymax}) 
          {\emph{unknown}};

          \node[below= 0mm of unknown,  xshift=-2.7mm, anchor=north, outer sep=0pt, inner sep=0pt] 
          {\footnotesize \emph{(solver timeout)}};

          \path[name path=top] (axis cs:\pgfkeysvalueof{/pgfplots/xmin},\pgfkeysvalueof{/pgfplots/ymax}) -- (axis cs:\pgfkeysvalueof{/pgfplots/xmax},\pgfkeysvalueof{/pgfplots/ymax});
          
          \path[name path=left] (axis cs: \pgfkeysvalueof{/pgfplots/xmin},\pgfkeysvalueof{/pgfplots/ymin}) -- (axis cs:\pgfkeysvalueof{/pgfplots/xmin},\pgfkeysvalueof{/pgfplots/ymax});
          \path[name path=bottom] (axis cs:\pgfkeysvalueof{/pgfplots/xmin},\pgfkeysvalueof{/pgfplots/ymin}) -- (axis cs:\pgfkeysvalueof{/pgfplots/xmax},\pgfkeysvalueof{/pgfplots/ymin});
          \path[name path=right] (axis cs:\pgfkeysvalueof{/pgfplots/xmax},\pgfkeysvalueof{/pgfplots/ymin}) -- (axis cs:\pgfkeysvalueof{/pgfplots/xmax},\pgfkeysvalueof{/pgfplots/ymax});

            \addplot[fill=black!30] 
            fill between[of=HU and right, soft clip={domain y=3:5}];
            \addplot[fill=black!30] 
            fill between[of=HU and HS];          
            \addplot[fill=black!30] 
            fill between[of=HS and top, soft clip={domain=2.6:3.5}];

            \addplot [fill=darkgreen!30] 
            fill between[of=bottom and HS,soft clip={domain=2.5:3.5}];
            \addplot [fill=darkgreen!30] 
            fill between[of=bottom and US,soft clip={domain=0:2.6}];

            \addplot [fill=darkred!30] 
            fill between[of=HU and top];
            \addplot [fill=darkred!30] 
            fill between[of=left and US];
            \addplot [fill=darkred!30]
            fill between[of=left and HU];

      \end{axis}
    \end{tikzpicture}

    \caption{Illustration of our asymptotic results for the power law satisfiability threshold location when $n\to\infty$ (left) compared with empirical results for randomly generated power law $3$-SAT formulas on $n=10^6$ variables checked with the SAT solver MiniSAT (right). The timeout was set to one hour.}

    \label{fig:illustration}

\end{figure}
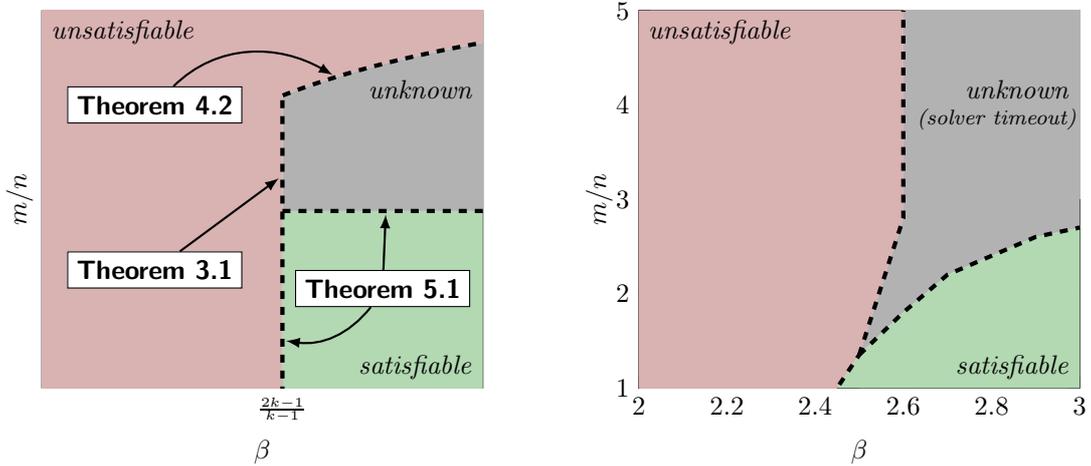

\medskip\noindent
{\bf Our results:}
We study random $k$-SAT on $n$ variables and $m=\Theta(n)$ clauses.
Each clause contains $k= \Theta(1)$ different, independently sampled variables. 
Each variable $x_i$ is chosen with non-uniform probability $p_i$ and
negated with probability $\nicefrac{1}{2}$. A formal definition
can be found in Section~\ref{sec:prelim}.
We first study sufficient conditions under which the resulting $k$-SAT
instances are \emph{unsatisfiable}. 
Assume a probability distribution~$\vec p$ on the variables where $p_i$ is non-decreasing in $i\in\{1,\ldots, n\}$.
If the $k$ most frequent variables are sufficiently common, we prove in
\secref{unsat} the following statement:
\begin{restatable*}{theorem}{statethmunsat} \label{thm:unsat}
Let $\Phi$ be a random $k$-SAT formula with probability distribution~$\vec p$ on the variables (c.f. \defref{randomSAT}),
with $k\geq2$ and $\tfrac mn = \Omega(1)$. If $p_{n-k+1} = \Omega( (\tfrac{\log n}{n})^{1/k})$,
then $\Phi$ is \whp\ unsatisfiable.
\end{restatable*}
Our focus are power law distributions with some exponent~$\beta$.
\thmref{unsat} implies that power law random $k$-SAT formulas with
$\beta=\tfrac{2k-1}{k-1} - \eps$ for an arbitrary constant $\varepsilon > 0$ are unsatisfiable with high probability\footnote{We say that an event $E$ holds {\em \whp}, if there exists a $\delta > 0$ such that $\Pr[E] \geq 1- \Oh(n^{-\delta})$.}, cf.~\corref{unsat}.

In Section~\ref{sec:flip} we show that something similar holds for the clause-variable ratio $\tfrac{m}{n}$, i.e. power law random $k$-SAT formulas with $\tfrac{m}{n}$ bigger than some constant are unsatisfiable with high probability.
Although this already follows from basic observations, we derive a better bound on the value of the constant.
\begin{restatable*}{theorem}{statethmflip}\label{thm:singleflip-general}
Let $\Phi$ be a random $k$-SAT formula with probability distribution~$\vec p$ on the variables (c.f. \defref{randomSAT}),
with $k\geq2$ and $r=\tfrac mn$.
With high probability, $\Phi$ is unsatisfiable if
\[\left(1-\tfrac{1}{2^k}\right)^r\left[\prod_{i=1}^{n}{\left[2-\left(1-\frac{k\cdot p_i}{2^k-1}\frac{1}{\left(1-\tfrac12 k^2 \|\vec p \|^2_2\right)}\right)^m\right]}\right]^{\frac1n}<1.\]
\end{restatable*}
In \secref{sat} we prove the following positive result, which complements our picture of the satisfiability landscape:
\begin{restatable*}{theorem}{statethmsat} \label{thm:sat}
Let $\Phi$ be a random $k$-SAT formula whose variable probabilities follow a power law distribution (c.f. \defref{general}). If the power law exponent is $\betabound$ for an arbitrary $\eps>0$, $\Phi$ is satisfiable with high probability if $\tfrac mn$ is a small enough constant.
\end{restatable*}
Together our main theorems prove that random $k$-SAT instances whose variables follow power law distributions do not only exhibit a phase transition for some clause-variable ratio $r=\tfrac mn$, but also around the power law exponent $\beta=\tfrac{2k-1}{k-1}$.
\figref{illustration} contains an overview of our results.
To prove these statements, we borrow tools developed for the uniform random SAT model. Note, however, that many of their common techniques like the differential equation method seem difficult to apply
to non-uniform distributions; as removing a variable results in a more complex rescaling of the rest of the distribution. It is therefore crucial to perform careful operations on the formulas that leave the distribution of variables intact. To this end, we use techniques known from the analysis of power law random graphs.

\medskip\noindent
{\bf Clause length:}
We focus on power law variable distributions but fix the length of every clause to $k \geq 2$.
Power law models have also been proposed in which clause length is distributed by a power law as well~\cite{ABL2009,AnsoteguiBL09}. As long as there is a constant \emph{minimum clause length} $k_{\min} \geq 2$, our results can be extended to this case in the following way.

If the clause lengths are distributed as a power law, there will appear $\Theta(n)$ clauses of length $k_{\min}$, and all other clauses are of larger size. In that case, \thmrefs{unsat}{sat} are directly applicable to the linear number of clauses with size $k_{\min}$ (obtaining different hidden constants); and we have that the formula is satisfiable with high probability if $\beta \geq \tfrac{2k_{\min} - 1}{k_{\min} -1} + \eps$ and $m/n$ is a small enough constant. On the other hand, the formula is unsatisfiable with high probability, if  $\beta \leq \tfrac{2k_{\min} - 1}{k_{\min} -1} - \eps$. Consequently, the satisfiability of the formula does (asymptotically) not depend on the second power law.

\section{Definition of the Model and Preliminaries}
\label{sec:prelim}
We analyze random $k$-SAT on $n$ variables and $m=\Theta(n)$ clauses, where $k \geq 2$. The constant $r := \tfrac mn$ is called {\em clause-variable ratio} or {\em constraint density}. We denote by $x_1, \ldots, x_n$ the Boolean variables. A clause is a disjunction of $k$ literals $\ell_1 \lor \ldots \lor \ell_k$, where each literal assumes a (possibly negated) variable. Finally, a formula $\Phi$ in conjunctive normal form is a conjunction of clauses $c_1 \land \ldots \land c_m$. We conveniently interpret a clause $c$ both as a Boolean formula and as a set of literals. Following standard notation, we write $|\ell|$ to refer to the indicator of the variable corresponding to literal $\ell$. We say that $\Phi$ is satisfiable if there exists an assignment of variables $x_1, \ldots, x_n$ such that the formula evaluates to $1$.

\begin{definition}[Random $k$-SAT]
\label{def:randomSAT} 
Let $m,n$ be given, and consider any probability distribution $\vec p$ on $n$ variables with $\sum_{i=1}^n p_i = 1$. To construct a random SAT formula $\Phi$, we sample $m$ clauses independently at random. Each clause is sampled as follows:

\begin{enumerate}
	\item Select $k$ variables independently at random from the distribution $\vec p$. Repeat until no variables coincide.
	\item Negate each of the $k$ variables independently at random with probability $\nicefrac{1}{2}$.
\end{enumerate}

\end{definition}
Observe that by setting $p_i = \tfrac 1n$ for all $i$, we obtain again the uniform random SAT model. 
The probability to draw a specific clause $c$ is
\begin{equation}
\frac{\prod_{\ell \in c} p_{|\ell|}}{2^k\sum_{J\in{\mathcal{P}_k\left(\left\{1,2,\ldots,n\right\}\right)}}{\prod_{j\in{J}}{p_j}}}, \label{eq:clause-prob}
\end{equation}
where $\mathcal{P}_k(\cdot)$ denotes the set of cardinality-$k$ elements of the power set.
The factor $2^k$ in the denominator comes from the different possibilities to negate variables. Note that $k!\sum_{J\in{\mathcal{P}_k\left(\left\{1,2,\ldots,n\right\}\right)}}{\prod_{j\in{J}}{p_j}}$ is the probability of choosing a $k$-clause that contains no variable more than once. 
To see that this probability is almost $1$ for most distributions, we apply the following result from \cite{alistarh2015lock}.
\begin{lemma}[Non-Uniform Birthday Paradox]
\label{lem:birthday}
Let $\vec p = (p_1, \ldots, p_n)$ be any probability distribution on $n$ items. Assume you sample $t$ items from $\vec p$. Let $\mathcal{E}(t)$ be the event that there is a collision, i.e.\ that at least 2 of $t$ items are equal. Then,
\[\Pr[\mathcal{E}(t)] \leq \tfrac12 t^2 \|\vec p \|^2_2 = \tfrac12 t^2 \sum_{i=1}^n p_i^2.\]
\end{lemma}
The probability that a sampled $k$-clause thereby has collisions is at most $\tfrac12 k^2 \|\vec p\|_2^2$; so for $\|\vec p\|_2^2 = o(1)$ and constant $k$ we obtain that the probability to draw a specific clause $c$ is
\begin{equation}
(1+o(1))\frac{k!}{2^k} \prod_{\ell \in c} p_{|\ell|}. \label{eq:clausesample}
\end{equation}

\medskip\noindent
{\bf Power law Distributions.}
In this paper, we are mostly concerned with distributions $p_i$ that follow a power law. To this end, we define two models: A {\em general} model to capture most power law distributions (which is harder to analyze), and a {\em concrete} model that gives us one instance of $\vec p$ depending only on $n$ that can be used to compute precise leading constants. We use the general model to derive some asymptotic results; and the concrete model to compare with the uniform random SAT model and for the experiments.

Before we define these two models, let us establish the concept of a {\em weight} $w_i$ of a variable~$x_i$. The weight gives us (roughly) the expected number of times the variable appears in the formula. That is, 
 \[p_i := \frac{w_i}{\sum_j w_j}.\]
Thus, fixing the weights $\vec w = (w_1, \ldots, w_n)$ also fixes the probability distribution $\vec p$. 
It is important to distinguish between the initial distribution of variables $\vec p$ and modified distributions that may arise as a result of stochastic considerations. For instance, the smallest-weight variable in a clause is clearly not distributed according to $\vec p$ (except in $1$-SAT). To avoid confusion, we identify a variable with its weight, as the weights stay fixed throughout the analysis. For convenience, we further assume \wlo that the variables are ordered increasingly by weight, \ie for $i \leq j$ we have $w_i \leq w_j$. Note that our definition of power law ensures that for $\beta>2$, we have $\sum_j w_j = \Theta(n)$.

We are now ready to define the two models.
\begin{definition}[General Power Law]
\label{def:general}
Let the weights $\vec w := w_1, \ldots, w_n$ be given, and let $W$ be a weight selected uniformly at random.
We say that $\vec w$ follows a power law with exponent $\beta$, if $w_1 = \Theta(1)$, $w_n = \Theta(n^{\frac{1}{\beta-1}})$, and for all $w \in [w_1, w_n]$ it holds
\begin{equation}
F(w) := \Pr[W \geq w] = \Theta(w^{1-\beta}) \label{eq:plw}
\end{equation}
\end{definition}
Whenever we need the explicit constants bounding the distribution function, we refer to them by $\alpha_1, \alpha_2$ as in
\begin{equation} \alpha_1 w^{1-\beta} \leq F(w) \leq \alpha_2 w^{1-\beta}. \label{eq:sandwhich} \end{equation}
We point out that \defref{general} assumes a deterministic weight sequence; but it can be easily generalized to also support randomly generated weights.

For the concrete model, we define the weights as follows.
\begin{definition}[Concrete Power Law]
\label{def:concrete}
Given a power law exponent $\beta$, we call $\vec w$ the concrete power law sequence, if
 \begin{equation}
\label{eq:exactplw}w_{n-i+1} := (\tfrac ni)^{\frac{1}{\beta-1}}.
\end{equation}
\end{definition}
One can check that for these concrete weights, it holds $n \cdot F(w) = \lfloor nw^{1 -\beta} \rfloor$, so in a sense, they are a canonical choice for producing a power law weight distribution.

To analyze power law distributions, we often make use of the following result of  Bringmann,
Keusch, and Lengler~\cite[Lemma~B.1]{BringmannKeuschLengler}, which allows replacing sums by integrals.
\begin{theorem}[\cite{BringmannKeuschLengler}]
\label{thm:karl}
Let $f\colon \mathbb R \rightarrow \mathbb R$ be a continuously differentiable function, and let $F^>(w) := \Pr[W > w]$. Then, for any $0 \leq \underline{w} \leq \bar w$,
\[\sum_{i\in [n], \underline w \leq w_i \leq \bar w} \tfrac1n f(w_i) = f(\underline w) \cdot F(\underline w)  - f(\bar w) \cdot  F^>(\bar w) + \int_{\underline w}^{\bar w} f'(w) \cdot F(w) \dif w. \]
\end{theorem}
Using this theorem, the following corollary can be shown: 
\begin{corollary}\label{cor:sumweights} Let the variables $w_i$ be power law distributed with exponent $\beta > 2$, and define $W_{\geq w} := \sum_{i\in[n]\colon w_i \geq w} w_i$.  Then, $W_{\geq w} = \Theta( nw^{2-\beta})$.
\end{corollary}
\begin{proof}
Observe that $\sum_{w' \geq w} w' = \sum_{i \in [n], w_i \geq w} w_i$.  We apply \thmref{karl} to obtain
\begin{align*}
\tfrac 1n \sum_{i \in [n], w_i \geq w} w_i &= w \cdot F(w) + \int_{w}^{w_n} F(v) \dif{} v \\
&\leq \alpha_2 w^{2-\beta} + [\tfrac{\alpha_2}{2-\beta} v^{2-\beta}]_w^{w_n} \\
&\leq \alpha_2\tfrac{\beta-1}{\beta-2} w^{2-\beta}.
\end{align*}
In a similar fashion, one may show that $\tfrac 1n \sum_{i \in [n], w_i \geq w} w_i \geq \alpha_1\tfrac{\beta-1}{\beta-2}(1-o(1)) w^{2-\beta}$.
\end{proof}

Hence, $\sum_j w_j = W_{\geq w_1} = \Theta(n)$ and therefore $p_i = \Theta(\tfrac{w_i}n)$.
Finally, we denote by $V$ the random variable describing the weight of a SAT variable chosen according to a power law distribution $p_i$, that is, $\Pr[V = w] = \sum_{i} p_i \cdot \mathds{1}[w_i = w]$, where $\mathds{1}$ denotes the indicator variable of the event. Note that this is not equivalent to $W$, since there is a subtle difference in the two random processes: $W$ is a random variable drawn uniformly at random from $w_1, \ldots, w_n$, whereas $V$ is a random variable drawn from the same set, but with the non-uniform distribution $p_1, \ldots, p_n$. Hence, by \corref{sumweights},
\begin{equation}
\Pr[V \geq w]
= \Theta(w^{2-\beta}) . \label{eq:varsample}
\end{equation}

Using \thmref{karl}, we can show that the probability to draw a certain clause $c$ is as given by \eq{clausesample} for \defref{general} with exponent $\beta>2$, since
\[\|\vec p\|^2_2 = \sum_{i=1}^n p_i^2 = \Theta(n^{-2}) \sum_{i=1}^n w_i^2 = \Theta(n^{-1}) \cdot n^{\frac{3-\beta}{\beta-1}} = o(1).\]

It remains to show that using a power law distribution in \defref{randomSAT} indeed results in a power law distribution of variable occurrences. Ansótegui et al.~\cite{AnsoteguiBL09} provide a proof sketch for this fact, we prove it rigorously.
\begin{theorem} \label{thm:frequencies}
Let $\Phi$ be a random $k$-SAT formula that follows an arbitrary power law distribution with exponent $\beta$ (c.f. \defref{general}) and $m=\Theta(n)$.
Then, there are \mbox{$d_{\min}=\Theta\left(w_{\min}\right)$} and $d_{\max}=\Theta\left(w_{\max}\right)$, such that for all $d_{\min}\le d\le d_{\max}$ \whp it holds that
\[N_{\ge d}=\Theta(n\cdot d^{1-\beta}),\] 
where $N_{\ge d}$ is the number of variables that appear at least $d$ times in $\Phi$.
\end{theorem}
\begin{proof}
Let $f_x$ be the number of appearances of $x$ or $\bar{x}$ in $\Phi$. Observe that $\Ex{f_x} \leq k\cdot m \cdot p_x$, since each variable can appear at most once in a clause. On the other hand, it holds $\Ex{f_x} \geq m \cdot p_x$, since this is the expected number of appearances of $x$ in a $1$-SAT formula. Thus, since $m = \Theta(n)$ and $k=\Theta(1)$ by assumption, it holds that
\begin{equation}
\Ex{f_x}=\Theta(w_x).\label{eq:exp-freq}
\end{equation}

We first prove the statement for $d\ge 2c\ln n$, where $c>0$ is some suitable constant.
By applying Chernoff bounds, we can derive that \whp all variables $x$ with $\Ex{f_x}< \tfrac d 2$ appear fewer than $d$ times; and all variables $x$ with $\Ex{f_x}\ge 2d$ appear at least $d$ times.
The requirement $d_{\min}\le d\le d_{\max}$ is needed so that the Chernoff bounds work, which might not be the case if $d$ is too close to $w_{\min}$ or $w_{\max}$.
Due to \defref{general} and \eq{exp-freq} this implies
\[N_{\ge d} = \Theta\left(n\cdot d^{1-\beta}\right).\]
Now let us consider the case $d< 2c\ln n$ and let $Y_i$ be random variables indicating if $f_{x_i}\ge d$ for $i=1,2,\ldots,n$.
To show a lower bound on $N_{\ge d}$, we again look at variables $x$ with $\Ex{f_x}\ge2d$. 
For those it holds that
\begin{equation}
\Pr\left(Y_i=0\right)\le\Pr\left(f_{x_i} < \frac12 \Ex{f_{x_i}}\right)\le e^{-\frac{\Ex{f_{x_i}}}8} \le e^{-\frac d 4},
\end{equation}
again due to Chernoff Bounds.
Also, by \eq{exp-freq}, it holds for variables $x$ with $\Ex{f_x}\ge2d$ that $w_x=\Omega(2d)$.
By the requirements on the weight distribution from \defref{general} there are $\Theta(n\cdot d^{1-\beta})$ such variables.
Therefore, it holds that
\[\Ex{N_{\ge d}}\ge\Ex{\sum_{\substack{i\in[n]\colon\\ \Ex{f_{x_i}}\ge2d}}Y_i}\ge c'\left(1-e^{-\frac d 4}\right)\cdot n \cdot d^{1-\beta}\]
for a suitable constant $c'>0$.
Observe that if we condition on $Y_i = 1$, \ie that $x_i$ appears at least $d$ times, this lowers the probability of all other variables to appear $d$ times, and vice versa. Thus, the random variables $Y_1,\ldots,Y_n$ are negatively correlated and we may apply a Chernoff bound~\cite[Theorem 1.16]{auger2011theory}.
Since $1 - e^{-d/4} = \Omega(1)$ and $d = \Oh(\log n)$, we obtain that \whp\
\[N_{\ge d} \ge \frac12 c'\cdot \left(1-e^{-\frac d 4}\right)\cdot n \cdot d^{1-\beta} = \Omega(n\cdot d^{1-\beta}).\]

To show an upper bound on $N_{\ge d}$ we consider variables $x$ with $\Ex{f_x}\le \tfrac d {2e}$ of which there are $n\cdot (1 -\Theta(d^{1-\beta}))$ due to \defref{general}.
For these variables, by a Chernoff bound~\cite{DP09} it holds that
\[\Pr\left(Y_i=1\right)\le\Pr\left(f_{x_i} > 2e\cdot \Ex{f_{x_i}}\right)\le 2^{-d}.\]
Now let $N'_{\ge d}$ be the number of variables with $\Ex{f_x}\le \frac d{2e}$ and $f_x\ge d$.
Thus, there exists a constant $c'' > 0$ such that
\[\Ex{N'_{\ge d}}\le n\cdot 2^{-d}\left(1- c'' \cdot d^{1-\beta} \right).\]
Due to negative association of the $Y_i$'s we can again use a Chernoff bound, yielding that \whp,
\begin{equation}
N'_{\ge d} \le n\cdot  2^{1-d}\left(1- c'' \cdot d^{1-\beta}\right).\label{eq:N'-chernoff}
\end{equation}
If $\Ex{N'_{\ge d}}$ is very small, for example $\Ex{N'_{\ge d}}=\Oh\left(\log n\right)$, then we can use negative association to apply the Chernoff bound with $t>2e\cdot\Ex{N'_{\ge d}}$ to achieve \eq{N'-chernoff} with high probability, since $t=n\cdot d^{1-\beta}=\Omega\left(\frac{n}{\operatorname{polylog}(n)}\right)$.
Observe that $1 - c'' d^{1-\beta} = \Oh(1)$. Furthermore, for variables $x$ with $\Ex{f_x}> \frac d{2e}$, we pessimistically assume $f_x\ge d$.
This gives us
\[N_{\ge d} = \Oh\left(n\cdot (d^{1-\beta}+2^{-d})\right) = \Oh\left(n\cdot d^{1-\beta}\right),\]
since $2^{-d}=\Oh\left(d^{1-\beta}\right)$ for constant $\beta$.
\end{proof}

\section{Small Power Law Exponents are Unsatisfiable}
\label{sec:unsat}
For small power law exponents, one can show that they result in formulas that are unsatisfiable (for large $n$) for all constant clause-variable ratios. The rationale behind this is that large variables with weight $\Theta(w_n)$ appear polynomially often together in a clause. For constant $k$, they thus appear in all $2^k$ configurations (negated and non-negated), making the formula trivially unsatisfiable. Theorem~\ref{thm:unsat}, already stated in the introduction, gives a sufficient condition on the variable distribution to make a random $k$-SAT formula unsatisfiable.
\statethmunsat
\begin{proof}
Recall that $p_i$ is without loss of generality increasing in $i$. Consider the $k$ largest variables $n-k+1, \ldots, n$.
We call $\mathcal{E}_i$ the event that clause $i$ consists of these variables. Then,
\begin{align*}
\Pr[\mathcal{E}_i] = \Omega(p_{n-k+1}^k) = \Omega(\tfrac{\log n}{n}).
\end{align*}
Since each clause is drawn independently at random, we obtain by a Chernoff bound (see for example Theorem~1.1 in \cite{DP09}) that with high probability, the total number of clauses consisting of these variables is \vspace{-2mm}
\[\vspace{-2mm} |\mathcal{E}| := \sum_{i=1}^m  \mathds{1}[\mathcal{E}_i] = \Omega(\log n).\]
In other words, the number of clauses in which the $k$ largest variables appear together increases as a logarithm in $n$. Since in each of these clauses, the literals appear negated or non-negated with constant probability $\nicefrac12$, we have that all $2^k$ possible combinations of negated and non-negated literals appear in the formula with probability at least
\[1 - 2^k \cdot (\tfrac{2^k-1}{2^k})^{|\mathcal E|} = 1 - n^{-\Omega(1)}\]
by the union bound.
Since all $2^k$ combinations cannot be satisfied at once, the resulting formula is unsatisfiable.
\end{proof}
By applying \thmref{unsat} to a power law distribution on the variables, we obtain the following power law threshold for unsatisfiability.
\begin{corollary}
\label{cor:unsat}
Let $\Phi$ be a random $k$-SAT formula that follows an arbitrary power law distribution fulfilling \defref{general}. If the power law exponent is $\beta \leq \tfrac{2k-1}{k-1} - \eps$
for an arbitrary $\eps>0$, $\Phi$ is unsatisfiable with high probability.
\end{corollary}
\begin{proof}
Observe that from $\beta = \tfrac{2k-1}{k-1} - \eps $ it follows $k = \tfrac{\beta-1}{\beta-2} - \eps'$ for some constant $\eps'$.
By setting $nF(w) \leq k$ we obtain that the largest $k$ variables all have weight 
$\Theta(w_n) = \Theta(n^{\frac{1}{\beta-1}}).$ 
Consequently, when $\beta > 2$,
\[(p_{n-k})^k = \Theta(n^{-k\frac{\beta-2}{\beta-1}}) = \Theta(n^{-1+\varepsilon'\frac{\beta-2}{\beta-1}}) = \omega(\tfrac{\log n}{n}),\]
and the statement follows from \thmref{unsat}. For the case where $\beta \leq 2$, one can show using \thmref{karl} that $\sum_i w_i = \Theta(n^{\frac{1}{\beta-1}})$, and therefore $p_{n-k} = \Omega(1)$. Again, the statement follows from \thmref{unsat}.
\end{proof}

\section{Large Clause-Variable Ratios are Unsatisfiable} 
\label{sec:flip}
It is a well-known result that random SAT on any probability distribution will result in unsatisfiable formulas if the clause-variable ratio is high. This follows from the probabilistic method: The expected number of assignments that satisfy a formula is $2^n (1-2^{-k})^m$. This is independent from the variable distribution as long as each variable is negated with probability $\nicefrac 12$. Hence, if the clause-variable ratio exceeds $\ln(2)/\ln(\tfrac{2^k}{2^k-1})$, the resulting formula will be unsatisfiable with high probability. This constant is rather large, however: In the case of $k=3$ this yields an upper bound on the clause-variable ratio of $\approx 5.191$. For the concrete power law distribution in \defref{concrete}, the true threshold is much smaller. In fact, it appears to be below the satisfiability threshold for uniform random SAT.

Let us restate the main result, which will be proven with the Single Flip Method \cite{KKKS98}.
\statethmflip
The following is a corollary from this theorem:
\begin{corollary}\label{cor:sfuniform}
Let $\Phi$ be a random $k$-SAT formula that follows \defref{randomSAT} with $k\geq2$, $r=\tfrac mn$ and $\|\vec p \|^2_2=o(1)$.
With high probability, $\Phi$ is unsatisfiable if
\[\left(1-\tfrac{1}{2^k}\right)^r\left(2-\exp\left(-\left(\frac{k}{2^k-1}r\right)\left(1+o(1)\right)\right)\right)<1.\]
\end{corollary}
\begin{proof}
We can upper-bound the left-hand side of the inequality as follows
\begin{align*}
& \left(1-\tfrac{1}{2^k}\right)^r\left[\prod_{i=1}^{n}{\left[2-\left(1-\frac{k\cdot p_i}{2^k-1}\frac{1}{\left(1-\tfrac12 k^2 \|\vec p \|^2_2\right)}\right)^m\right]}\right]^{\frac1n}\ \\
& \le \left(1-\tfrac{1}{2^k}\right)^r\left[\frac1n\sum_{i=1}^{n}{\left[2-\left(1-\frac{k\cdot p_i}{2^k-1}\frac{1}{\left(1-\tfrac12 k^2 \|\vec p \|^2_2\right)}\right)^m\right]}\right]\\
& =  \left(1-\tfrac{1}{2^k}\right)^r\left[2-\frac1n\sum_{i=1}^{n}{\left(1-\frac{k\cdot p_i}{2^k-1}\frac{1}{\left(1-\tfrac12 k^2 \|\vec p \|^2_2\right)}\right)^m}\right]
\end{align*}
by applying the inequality of arithmetic and geometric means.
Since $\frac{k\cdot p_i}{2^k-1}\frac{1}{\left(1-\tfrac12 k^2 \|\vec p \|^2_2\right)}$ is upper-bounding a probability, we can assume it to be at most $1$.
It now holds that
\begin{align*}
\left(1-\frac{k\cdot p_i}{2^k-1}\frac{1}{\left(1-\tfrac12 k^2 \|\vec p \|^2_2\right)}\right)^m
	&	> \exp\left(-\frac{m}{\frac{\left(2^k-1\right)\left(1-\tfrac12 k^2 \|\vec p \|^2_2\right)}{k\cdot p_i}-1}\right)\\
	& = \exp\left({-\frac{k\cdot p_i}{2^k-1}\frac{m}{\left(1-\tfrac12 k^2 \|\vec p \|^2_2\right)}(1+o(1))}\right),
\end{align*}
since $\|\vec p \|^2_2=o(1)$ implies $\max_i\left(p_i\right)=o(1)$.
By plugging this into the inequality from before and applying the inequality of arithmetic and geometric means again, we get
\begin{align*}
& \left(1-\tfrac{1}{2^k}\right)^r\left[\prod_{i=1}^{n}{\left[2-\left(1-\frac{k\cdot p_i}{2^k-1}\frac{1}{\left(1-\tfrac12 k^2 \|\vec p \|^2_2\right)}\right)^m\right]}\right]^{\frac1n}\ \\
& \le \left(1-\tfrac{1}{2^k}\right)^r\left[2-\frac1n\sum_{i=1}^{n}{\exp\left(-\left(\frac{k\cdot p_i}{2^k-1}\frac{m}{\left(1-\tfrac12 k^2 \|\vec p \|^2_2\right)}\right)\left(1+o(1)\right)\right)}\right]\\
& \le \left(1-\tfrac{1}{2^k}\right)^r\left[2-\left(\prod_{i=1}^{n}{\exp\left(-\left(\frac{k\cdot p_i}{2^k-1}\frac{m}{\left(1-\tfrac12 k^2 \|\vec p \|^2_2\right)}\right)\left(1+o(1)\right)\right)}\right)^{\frac1n}\right]\\
& = \left(1-\tfrac{1}{2^k}\right)^r\left[2-\exp\left(-\left(\frac{k}{2^k-1}\frac{r}{\left(1-\tfrac12 k^2 \|\vec p \|^2_2\right)}\right)\left(1+o(1)\right)\right)\right].
\end{align*}
For $\|\vec p \|^2_2=o(1)$ this is roughly
\[\left(1-\tfrac{1}{2^k}\right)^r\left(2-\exp\left(-\left(\frac{k}{2^k-1}r\right)\left(1+o(1)\right)\right)\right).\qedhere \]
\end{proof}

Interestingly, the above corollary gives the same inequality as the Single-Flip Method for uniform random SAT~\cite{KKKS98}.
This shows that the uniform distribution resembles a worst-case for this method; and all other distributions can only improve this bound.

If $\vec p$ follows a power law distribution as in \defref{concrete}, we can derive the following theorem, which gives an upper bound independent of $n$.
\begin{theorem}\label{thm:singleflip-powerlaw}
Let $\Phi$ be a random $k$-SAT formula with $k\geq2$ and $r=\tfrac mn$ that follows a power law distribution fulfilling \defref{concrete}. Let further $N\in\mathbb{N}^+$ be any constant. If the power law exponent is $\beta>2$, then $\Phi$ is \whp\ unsatisfiable if 
\[\left(\left(1-\tfrac{1}{2^k}\right)^r 2^{\frac 1N}\prod_{l=1}^{N-1}\left[2-\exp{\left(-\left(1+o{(1)}\right)r\frac{k}{2^k-1}\frac{\beta-2}{\beta-1}\left(\frac{N}{l}\right)^{\frac{1}{\beta-1}}\right)}\right]^{\frac 1N}\right)<1.\]
\end{theorem}
\begin{proof}
We apply \thmref{singleflip-general}. If $\vec p$ follows a power law distribution as in \defref{concrete}, we can further simplify $\Ex{N_{SF}}$ to
\begin{align*}
\Ex{N_{SF}} & \leq  \left(1-\tfrac{1}{2^k}\right)^m\prod_{i=1}^{n}{\left[2-\exp{\left(-\frac{\frac{k\cdot p_i}{2^k-1}m}{\left(1-\tfrac12 k^2 \|\vec p \|^2_2\right)-\frac{k\cdot p_i}{2^k-1}}\right)}\right]}
\end{align*}
using the inequality $1-x\geq e^{-\frac{x}{1-x}}$ which holds for all $x<1$. We upper bound the probabilities $p_i$ by choosing an integer $N\ge2$ and dividing the set of variables into $N$ buckets of equal size. For $i\in \left[\left\lceil\frac{l-1}{N} n\right\rceil  + 1, \left\lceil\frac{l}{N}n\right\rceil \right]$ and $1\le l\le N-1$ we can estimate
\[p_i \le  \frac{\left(\frac{n}{n-\left\lceil\frac{l}{N}n\right\rceil+1}\right)^{\frac{1}{\beta-1}}}{\sum_{i=1}^{n}{\left(\frac{n}{i}\right)^{\frac{1}{\beta-1}}}} = \frac{\left(\frac{N}{N-l}\right)^{\frac{1}{\beta-1}}}{\sum_{i=1}^{n}{\left(\frac{n}{i}\right)^{\frac{1}{\beta-1}}}}.\]
The last bucket $i\in \left[\left\lceil\frac{N-1}{N} n\right\rceil  + 1, n\right]$ is simply bounded by $2^{\frac nN}$ in the overall product.
W.l.o.g. we assume that there are exactly $\frac nN$ variables in each bucket, as  we could split the factor for $\left\lceil\frac{l}{N}n\right\rceil$ into appropriate parts which both obey the upper bound on $p_i$ for bucket $l$ and $l+1$ respectively. We can now upper-bound $\Ex{N_{SF}}$ by
\[\left(1-\tfrac{1}{2^k}\right)^m 2^{\frac nN}\prod_{l=1}^{N-1}\left[2-\exp{\left(-\frac{\frac{k}{2^k-1}\cdot m\cdot \left(\frac{N}{l}\right)^{\frac{1}{\beta-1}}}{\left(\sum_{i=1}^{n}{\left(\frac{n}{i}\right)^{\frac{1}{\beta-1}}}\right)\left(1-\tfrac12 k^2 \|\vec p \|^2_2\right)-\frac{k}{2^k-1} \left(\frac{N}{l}\right)^{\frac{1}{\beta-1}}}\right)}\right]^{\frac nN}
\]
We are now interested in what happens to the expression in the exponent when $n$ tends to infinity.
First, $\sum_{i=1}^{n}{\left(\tfrac{n}{i}\right)^{\frac{1}{\beta-1}}}\rightarrow \tfrac{\beta-1}{\beta-2}n$ for $\beta>2$.
Second, by \thmref{karl} we have that $\|\vec p\|^j_j = \Theta\left(n^{-j\frac{\beta-2}{\beta-1}}\right) \rightarrow 0$ whenever $j > \beta-1$ and $\|\vec p\|^j_j = \Theta\left(n^{-j+1}\right) \rightarrow 0$ whenever $j\le \beta-1$.
Finally, for every constant $N$ we have that $\frac{k}{2^k-1}\left(\tfrac{N}{l}\right)^{\frac{1}{\beta-1}}$ is also constant.
Using $m=r\cdot n$ we can thus simplify
\[\frac{\frac{k}{2^k-1}\cdot m\cdot \left(\frac{N}{l}\right)^{\frac{1}{\beta-1}}}{\left(\sum_{i=1}^{n}{\left(\frac{n}{i}\right)^{\frac{1}{\beta-1}}}\right)\left(1-\tfrac12 k^2 \|\vec p \|^2_2\right)-\frac{k}{2^k-1} \left(\frac{N}{l}\right)^{\frac{1}{\beta-1}}}
=
\left(1+o{(1)}\right)r\frac{k}{2^k-1}\frac{\beta-2}{\beta-1}\left(\frac{N}{l}\right)^{\frac{1}{\beta-1}}.\]
Plugging this into our inequality we get
\begin{align*}
\Ex{N_{SF}} & \leq  \left(1-\tfrac{1}{2^k}\right)^m 2^{\frac nN}\prod_{l=1}^{N-1}\left[2-\exp{\left(-\left(1+o{(1)}\right)r\frac{k}{2^k-1}\frac{\beta-2}{\beta-1}\left(\frac{N}{l}\right)^{\frac{1}{\beta-1}}\right)}\right]^{\frac nN}\\
& = \left(\left(1-\tfrac{1}{2^k}\right)^r 2^{\frac 1N}\prod_{l=1}^{N-1}\left[2-\exp{\left(-\left(1+o{(1)}\right)r\frac{k}{2^k-1}\frac{\beta-2}{\beta-1}\left(\frac{N}{l}\right)^{\frac{1}{\beta-1}}\right)}\right]^{\frac 1N}\right)^n
\end{align*}
This establishes \thmref{singleflip-powerlaw}.
\end{proof}

The bound from this Theorem improves as $N\to\infty$.
As this expression is rather terse, we also numerically determine in \tabref{flip} the smallest constant $r$ such that the formula is unsatisfiable. 
We compare these values to the upper bounds for uniform random SAT obtained from the Single-Flip Method.
\begin{table}[t]
\begin{center}
\begin{tabular}{@{\ }rrrrrrrrrc@{\ \ }r@{\hspace*{.3cm}}}
\toprule
& \multicolumn{8}{ c }{\bf power law distribution with exponent~$\bm\beta$}
& & \multirow{2}{*}{\begin{minipage}{2cm}\begin{center}\bf uniform\\dist.\end{center}\end{minipage}\!\!\!\!\!\!}
\\ \cmidrule{2-9}
$\bm k$							& \textbf{2.2} 	& \textbf{2.3}		& \textbf{2.4} 		& \textbf{2.5} 		& \textbf{2.6} 		& \textbf{2.7} 		& \textbf{2.8} 		& \textbf{2.9} 		&& \\\midrule
\textbf{3} 		&   						&  								&  								& 3.48 						& 3.71 						& 3.87 						& 3.99 						& 4.08 						&& 4.67\\
\textbf{4} 		&   						&  								& 7.87 						& 8.42 						& 8.78 						& 9.04 						& 9.23 						& 9.37 						&& 10.23\\
\textbf{5} 		&   						& 16.27 					& 17.75 					& 18.64 					& 19.21 					& 19.61 					& 19.90 					& 20.11 					&& 21.33\\
\textbf{7}		&  67.21 				& 75.74 					& 79.81 					& 82.09						& 83.49	 					& 84.42 					& 85.07 					& 85.54 					&& 87.88\\
\textbf{10} 	&  619.28 			& 662.48 					& 680.93 					& 690.36 					& 695.77 					& 699.12 					& 701.34 					& 702.88 					&& 708.94\\\addlinespace
\bottomrule
\end{tabular}
\end{center}
\caption{Numerical upper bounds on the density threshold obtained from the Single-Flip Method (cf.~\thmrefs{singleflip-general}{singleflip-powerlaw}). Empty fields indicate unsatisfiability for {\em all} constant densities by \thmref{unsat}.
To the best of our knowledge, the bounds for uniform random SAT with $k\ge 4$ are the currently best known numerical upper bounds.
For $k=3$ the best known unconditional numerical upper bound is 
$4.4898$ \protect\cite{DBLP:journals/tcs/DiazKMP09}.}
\vspace*{-6mm}
\label{tab:flip}
\end{table}
In the remainder of this section, we show \thmref{singleflip-general} and \thmref{singleflip-powerlaw}.
\begin{definition}[Single-Flip Property]
For a random formula $\Phi$ a truth assignment $A$ has the \emph{single-flip property} iff $A$ satisfies $\Phi$ and every assignment $A'$ obtained from $A$ by flipping exactly one zero to one does \emph{not} satisfy $\Phi$.
\end{definition}

Let $N_{SF}$ be the number of truth assignments with the single-flip property for $\Phi$. As argued in \cite{KKKS98}, such an assignment exists if $\Phi$ is satisfiable.
From Markov's Inequality, we thus know 
$\Pr[\Phi\text{ satisfiable}] \le \Ex{N_{SF}}.$ 

In the following, we derive a bound on $\Ex{N_{SF}}$.
By \lemref{birthday},
the probability of choosing a clause $c$ is at most
\[\frac{k!}{2^k} \cdot \frac{\prod_{\ell \in c}{p_{\left|\ell\right|}}}{1-\tfrac12 k^2 \|\vec p \|^2_2}.\]
To bound the number of assignments with the single-flip property, we use the following result.
\begin{lemma}[\cite{KKKS98}] \label{lem:expectation1}
The expected number of assignments with the single-flip property is
\[\Ex{N_{SF}}=\left(1-\tfrac{1}{2^{k}}\right)^m\sum_{\text{assignment }A}{\Pro{A\text{ single-flip} \mid A\text{ satisfying}}}.\]
\end{lemma}

\begin{proof}
Note that for a certain truth assignment $A$, the probability of choosing a clause which is not satisfied by $A$ is $\nicefrac{1}{2^k}$.
Therefore, the probability that $A$ is a satisfying assignment for $\Phi$ is exactly $\left(1-\tfrac{1}{2^k}\right)^m$.
\end{proof}
We next bound the probability that a satisfying assignment $A$ has the single-flip property.
\begin{lemma} \label{lem:conditional}
	For a satisfying assignment $A=\left(a_1, a_2, \ldots, a_n\right)\in{\left\{0,1\right\}^n}$ it holds that
	\begin{eqnarray*}
	\Pro{A\text{ single-flip}\ |\ A\text{ satisfying}}
	\le \prod_{i:\ a_i=0}{1-\left(1-\frac{k\cdot p_i}{2^k-1}\frac{1}{\left(1-\tfrac12 k^2 \|\vec p \|^2_2\right)}\right)^m}.
	\end{eqnarray*}
\end{lemma}

\begin{proof}
For a satisfying assignment $A$ to have the single-flip property, all assignments $A^i$ obtained by flipping a bit $a_i=0$ of $A$ must not satisfy $\Phi$.
To fulfill this property for $A^i$, we have to choose at least one clause which contains $\bar{X_i}$ and $k-1$ other variables with appropriate signs so that $A^i$ does not satisfy the clause. Let $S^i (c)$ denote the event that a clause $c$ is satisfied by $A$, but not by $A^i$. Then,
\begin{align*}
\Pr[S^i(c)] &= \frac{k!\cdot p_i\sum_{J\in{\mathcal{P}_{k-1}\left(\left[n\right]\setminus{\left\{i\right\}}\right)}}{\prod_{j\in J}p_j}}{2^k\left(1-\tfrac12 k^2 \|\vec p \|^2_2\right)} 
\le\frac{k\cdot p_i}{2^k\left(1-\tfrac12 k^2 \|\vec p \|^2_2\right)}
\end{align*}
since $\sum_{J\in{\mathcal{P}_{k-1}\left(\left[n\right]\setminus{\left\{i\right\}}\right)}}{\prod_{j\in J}p_j}\le\frac{\|\vec p \|^{k-1}_1}{(k-1)!}$.
The probability of choosing a clause not satisfied by $A^i$ under the condition that $A$ is satisfying is then
\[
\Pr[S^i(c) \mid A \text{ sat}] = \Pr[S^i(c) \mid A \text{ satisfies } c] 
\leq \frac{k\cdot p_i}{2^k-1}\frac{1}{\left(1-\tfrac12 k^2 \|\vec p \|^2_2\right)}
\]
as the probability of choosing a clause which is satisfied by any assignment is exactly $\frac{2^k-1}{2^k}$.
For a fixed assignment $A^i$ we conclude
\begin{align}
\Pro{A^i\text{ unsat}\ |\ A\text{ sat}}&= 1 -  \left( 1 - \Pr[S^i(c) \mid A \text{ sat}] \right)^m \nonumber\\
&\le 1-\left(1-\frac{k\cdot p_i}{2^k-1}\frac{1}{\left(1-\tfrac12 k^2 \|\vec p \|^2_2\right)}\right)^m. \label{eq:flipped-conditional}
\end{align}
It remains to find the joint probability that all single-flipped assignments $A^i$ for $1\le i \le n$ with $a_i=0$ are not satisfying.
We show this using a correlation inequality by Farr \cite{mcdiarmid1992}.
The sets of clauses which are not satisfied by the $A^i$'s are pairwise disjoint as each clause in the set for $A^i$ has to contain $\bar{X_i}$, whereas each clause in the set for $A^{j}$ ($j\neq i$) can not contain $\bar{X_i}$. In the context of the correlation inequality from \cite{mcdiarmid1992} we set $V=\left\{1,2,\ldots,m\right\}$, $I=\left\{i\in\left\{1,2,\ldots,n\right\}\ |\ a_i=0\right\}$, $X_v=i$ iff the $v$-th clause is satisfied by $A$, but not by $A^i$, and $\mathcal{F}_i$ the ``increasing'' collection of non-empty subsets of $V$.
The application of the Theorem then directly yields
\begin{align*}
\Pr[A \text{ single-flip} \mid A \text{ sat}] &= \Pro{\bigcap_{i:\ a_i=0} A^i\text{ unsat}\ |\ A\text{ sat}}  \\
&\le \prod_{i:\ a_i=0}{\left[1-\left(1-\frac{k\cdot p_i}{2^k-1}\frac{1}{\left(1-\tfrac12 k^2 \|\vec p \|^2_2\right)}\right)^m\right]}. \qedhere
\end{align*}
\end{proof}

Combining \lemrefs{expectation1}{conditional} we get that the expected number of assignments with single-flip property is at most
\begin{eqnarray*}
\Ex{N_{SF}} & \le & \left(1-\tfrac{1}{2^k}\right)^m\sum_{I\subseteq \left\{1,2,\ldots,n\right\}}{\prod_{i\in I}{\left[1-\left(1-\frac{k\cdot p_i}{2^k-1}\frac{1}{\left(1-\tfrac12 k^2 \|\vec p \|^2_2\right)}\right)^m\right]}}\\
& = & \left(1-\tfrac{1}{2^k}\right)^m\prod_{i=1}^{n}{\left[2-\left(1-\frac{k\cdot p_i}{2^k-1}\frac{1}{\left(1-\tfrac12 k^2 \|\vec p \|^2_2\right)}\right)^m\right]}.
\end{eqnarray*}
This establishes \thmref{singleflip-general}.

\section{Conditions for Satisfiability}
\label{sec:sat}

In this section, we provide a complementary result to \thmrefs{unsat}{singleflip-powerlaw} proving that if $\betabound$ and the clause-variable ratio $r = \tfrac mn$ does not exceed some small constant, then a random $k$-SAT formula with exponent $\beta$ is satisfiable with high probability. Let us first restate the main result:
\statethmsat
We show this statement by constructing an algorithm that satisfies $\Phi$ \whp\ if the clause-variable ratio is small. \algref{simple} contains a formal description. The main idea is to shrink all clauses to size $2$ by selecting the literals with smallest weight in each clause; and then running any well-known (polynomial time) 2-SAT algorithm (\eg\ \cite{aspvall1979linear}).

\renewcommand{\algorithmicrequire}{\textbf{Input:}}
\renewcommand{\algorithmicensure}{\textbf{Output:}}

\begin{algorithm}[t]
\caption{Clause Shrinking Algorithm}
\label{alg:simple}
\begin{algorithmic}[1]
\Require $k$-SAT formula $\Phi$; weight distribution $\vec w$
\ForAll{$c \in \Phi$}
  \State $\ell_1 \leftarrow \argmin_{\ell \in c} \{ w_{|\ell|} \}$
  \State $\ell_2 \leftarrow \argmin_{\ell \in c \setminus \{\ell_1 \} } \{ w_{|\ell|} \}$
  \State $c \leftarrow (\ell_1 \lor \ell_2)$
\EndFor
\State Solve $\Phi$ using any polynomial time 2-SAT algorithm 
\end{algorithmic}
\end{algorithm}

In the following, we seek to establish that \algref{simple} will find a satisfying assignment (for small constraint densities) with high probability. To this end, we first analyze the probability distribution of a clause $c$ after it has been shrunk.

\begin{lemma}
\label{lem:shrinked}
Let $\ell_1, \ell_2$ be the selected literals of an arbitrary clause $c \in \Phi$ in \algref{simple}. Then, 
\[ \Pr[|\ell_1| = i, |\ell_2|=j] + \Pr[|\ell_1| = j, |\ell_2|=i] \leq \Oh(\tfrac1{n^2} (w_iw_j)^{1-\frac12(k-2)(\beta-2)}). \]
\end{lemma}
\begin{proof}
\Wlog, we assume that $w_i \leq w_j$. Then, $\Pr[|\ell_1| = j, |\ell_2| = i] = 0$ by the definition of \algref{simple}.
For the event $|\ell_1| = i,$ $|\ell_2|=j$ to happen, all other $k-2$ literals in the clause must be of larger weight. By \eqs{clausesample}{varsample},
\begin{align*}
\Pr[|\ell_1| = i, |\ell_2| = j] &= \frac12 \cdot \binom{k}{2} \cdot (1 + o(1))  \cdot p_i \cdot p_j \cdot \Pr[V \geq w_j]^{k-2} \\
&= \Theta(\tfrac{1}{n^2}) \cdot w_iw_j^{1- (k-2)(\beta-2)} \\
&\leq \Oh(\tfrac1{n^2}) \cdot  (w_iw_j)^{1-\frac12(k-2)(\beta-2)}.
\end{align*}
The last statement holds since $w_i \leq w_j$.
\end{proof}
Having derived a bound on the probability distribution of a shrunk clause, it is possible to compute the probability that the resulting $2$-SAT formula is satisfiable. We use that the clauses are sampled independently. To avoid confusion, we write $\Phi'$ and $c'$, whenever we talk about the shrunk formula and clauses. To upper bound the probability of $\Phi$ not being satisfiable, we look at so-called {\em bi-cycles} in $\Phi'$.

\begin{definition}
A {\em bi-cycle} of length $l$ is a sequence of $l+1$ clauses of the form 
\[\left(u,\ell_1\right),\left(\bar{\ell}_1,\ell_2\right),\ldots,\left(\bar{\ell}_{l-1},\ell_l\right),\left(\bar{\ell}_l,v\right),\]
where $\ell_1,\ldots,\ell_l$ are literals of distinct variables and $u,v\in\left\{\ell_1,\ldots,\ell_l,\bar{\ell}_1,\ldots,\bar{\ell}_l\right\}$.
\end{definition}
Chvatal and Reed~\cite[Theorem~3]{chvatalreed92} show that if the formula $\Phi'$ is unsatisfiable, it must contain a bi-cycle. Consequently, by upper bounding the probability that a bi-cycle appears, we immediately obtain an upper bound on the probability that $\Phi'$ and henceforth $\Phi$ is unsatisfiable. 
\begin{theorem}[\cite{chvatalreed92}]
Let $\Phi'$ be any $2$-SAT formula. If $\Phi'$ contains no bi-cycle, it is satisfiable.
\end{theorem}
Before we are able to prove the main Theorem, we need the following auxiliary Lemma.
\begin{lemma}\label{lem:sum}
Let $\beta = \delta + 1 + \eps$ for some $\eps > 0$. For all $1 \le l \le n$, there is a constant $c$ with 
\begin{equation*} \sum_{\substack{S\subseteq\left[n\right]\colon\\ \left|S\right|=l}}{\prod_{i\in S}{w_i^\delta}}\le n^{l} \cdot c^{l} \tfrac1{l!}. \end{equation*}
\end{lemma}
\begin{proof}
We begin by observing that the term on the left side of the equation is obviously monotone in $w_i$: If $\delta \geq 0$ ($\delta < 0$), then increasing (decreasing) $w_i$ increases the sum. Thus, instead of considering the true distribution function $F(w)$, we may consider the upper (lower) bound on $F(w)$, see \eq{sandwhich}. For the sake of brevity, we consider the distribution $\widehat F(w) = \alpha w^{1-\beta}$, where $\alpha$ is chosen to be either $\alpha_1$ if $\delta \geq 0$, or $\alpha_2$ otherwise.

To estimate this sum, we arrange the elements of $S$ increasingly by weight, such that $w_{s_1} < w_{s_2} < \ldots < w_{s_l}$.
This gives us
\[\sum_{\substack{S\subseteq\left[n\right] \\ \left|S\right|=l}}{\prod_{i\in S}{w_i^\delta}} = \sum_{s_1=1}^{n-l+1}{\left(w_{s_1}^\delta\sum_{s_2=s_1+1}^{n-l+2}{\left(w_{s_2}^\delta\ldots\sum_{s_l=s_{l-1}+1}^{n}{w_{s_l}^\delta}\right)}\right)}.\]
We are now inductively estimating these sums, beginning with the innermost. Recall that $\widehat F(w) = \alpha w^{1-\beta}$. Let $d$ be a large enough constant. We establish the following induction hypothesis:
\begin{multline*}
\sum_{s_{l-i} = s_{l-i-1} + 1 }^{ n-i } \left(w_{s_{l-i}}^\delta \sum_{s_{l-i+1}= s_{l-i} + 1}^{n-i+1} \left(w_{s_{l-i+1}}^\delta\ldots\sum_{s_l=s_{l-1}+1}^{n}w_{s_l}^\delta\right)\right) \le \\
n^{i+1} \cdot w_{s_{l-i-1}}^{(i+1) \cdot (\delta+1-\beta)} \cdot d^{i+1} \tfrac1{(i+1)!}
\end{multline*}
Now we apply \thmref{karl} to prove the induction basis. For $i=0$, we have
\begin{align*}
\sum_{s_l=s_{l-1}+1}^{n}w_{s_l}^\delta &\leq n \cdot \alpha w_{s_{l-1}}^{\delta +1-\beta} + n \int_{w_{s_{l-1}}}^{w_n} \alpha \delta w^{\delta-\beta} \dif w\\
&\leq n \cdot \alpha w_{s_{l-1}}^{\delta +1-\beta} + n \cdot \tfrac{\alpha \delta}{\beta-\delta-1} w_{s_{l-1}}^{\delta+1-\beta} \\
&= n \cdot \tfrac{\alpha(\beta-1)}{\beta-\delta-1} w_{s_{l-1}}^{\delta +1-\beta} 
\end{align*}
as desired, since $\beta>\delta +1$.

Now suppose the induction hypothesis holds for $i-1$.
For $i$ we get
\begin{multline}
\sum_{s_{l-i} = s_{l-i-1} + 1 }^{ n-i } \left(w_{s_{l-i}}^\delta \sum_{s_{l-i+1}= s_{l-i} + 1}^{n-i+1} \left(w_{s_{l-i+1}}^\delta\ldots\sum_{s_l=s_{l-1}+1}^{n}w_{s_l}^\delta\right)\right) \le \\  
n^{i} \cdot d^{i} \tfrac1{i!} \sum_{s_{l-i} = s_{l-i-1} + 1 }^{ n-i } w_{s_{l-i}}^{i \cdot (\delta+1-\beta) + \delta} \label{eq:induction}
\end{multline}
To bound the sum, we distinguish two cases. If $i(\delta + 1 - \beta) + \delta \geq 0$, then
\begin{align*}
&\tfrac1n \sum_{s_{l-i} = s_{l-i-1} + 1 }^{ n-i } w_{s_{l-i}}^{i \cdot (\delta+1-\beta) + \delta} \\
\leq{}& \alpha w_{s_{l-i-1}}^{(i+1) \cdot (\delta+1-\beta)} + \int_{w_{s_{l-i-1}}}^{w_n} \alpha (i \cdot (\delta+1-\beta) + \delta) w^{i \cdot (\delta+1-\beta) + \delta - \beta} \dif w  \\
={}& \alpha w_{s_{l-i-1}}^{(i+1) \cdot (\delta+1-\beta)} + \left[\tfrac{\alpha (i \cdot (\delta+1-\beta) + \delta)}{(i+1) \cdot (\delta+1-\beta)} w^{(i+1) \cdot (\delta+1-\beta)} \right]_{w_{s_{l-i-1}}}^{w_n}  \\
={}& \left( \alpha + \tfrac{\alpha (i \cdot (\delta+1-\beta) + \delta)}{(i+1)(\beta-\delta-1)} \right) w_{s_{l-i-1}}^{(i+1) \cdot (\delta+1-\beta)} - \tfrac{\alpha (i \cdot (\delta+1-\beta) + \delta)}{(i+1)(\beta-\delta-1)} w_n^{(i+1) \cdot (\delta+1-\beta)} \\
={}& \tfrac{\alpha (\beta-1)}{(i+1)(\beta-\delta-1)} w_{s_{l-i-1}}^{(i+1) \cdot (\delta+1-\beta)} \cdot \left( 1 - \tfrac{i \cdot (\delta + 1 - \beta) + \delta}{\beta-1} (\tfrac{w_n}{w_{s_{l-i-1}}})^{(i+1)\cdot (\delta + 1 - \beta)} \right).
\end{align*}
For the integration, we need to make sure that the special case $-1 = i \cdot (\delta+1-\beta) + \delta - \beta$ does not occur. By rearranging, one can see that this is only true for $i = -1$, however, our weights begin at $i=1$. 

We now bound the error term that appears from the integration limit $w_n$. Note that we only need to consider the case where $i (\delta + 1 - \beta) < - \delta$, otherwise the error term is smaller than $1$ and may simply be omitted.

Observe from \eq{induction} that $w_{s_{l-i-1}} \leq w_{n-i}$. Further, by \eq{sandwhich} we have that $w_n = \Theta(n^{\frac{1}{\beta-1}})$. Similarly,  
\begin{equation}
\tfrac in = \widehat F(w_{n-i}) = \alpha w_{n-i}^{1-\beta},  
\end{equation}
therefore $w_{n-i} = \Theta(1) \cdot (\tfrac{n}{i})^{\frac1{\beta-1}}$. Therefore, we have
\begin{equation}
\frac{w_n}{w_{s_{l-i-1}}} \leq \frac{w_n}{w_{n-i}} = \Theta(i^{\frac{1}{\beta-1}}).
\end{equation}
Recall that we are in the case where the error term is positive; and that the exponent $(\delta + 1 -\beta)$ is negative. Substituting the above inequality, we obtain
\begin{align*}
1 - \tfrac{i \cdot (\delta + 1 - \beta) + \delta}{\beta-1} (\tfrac{w_n}{w_{s_{l-i-1}}})^{(i+1)\cdot (\delta + 1 - \beta)} 
&\leq 1 - \tfrac{i \cdot (\delta + 1 - \beta)}{\beta-1} i^{\frac{i+1}{\beta-1} \cdot (\delta + 1 - \beta)} \\
&= 1 - \tfrac{\delta+1-\beta}{\beta-1} i^{1 + \frac{i+1}{\beta-1} \cdot (\delta+1-\beta)}
\end{align*}
By inspecting the exponent $1 + \tfrac{i+1}{\beta-1} \cdot (\delta + 1 - \beta)$, we observe that it is of order $\Oh(1)$. In particular, once $i$ is a large enough constant, the exponent becomes negative. Therefore, we may conclude that
\begin{equation}
1 - \tfrac{\delta+1-\beta}{\beta-1} i^{1 + \frac{i+1}{\beta-1} \cdot (\delta+1-\beta)} = \Oh(1),
\end{equation}
where the constant is not dependent on the iteration $i$. Thus, as $d$ was chosen large enough, 
\begin{equation}
\tfrac1n \sum_{s_{l-i} = s_{l-i-1} + 1 }^{ n-i } w_{s_{l-i}}^{i \cdot (\delta+1-\beta) + \delta} \leq \tfrac{d}{i+1} w_{s_{l-i-1}}^{(i+1) \cdot (\delta + 1 - \beta)},
\end{equation}
Plugging this into inequality~\eqref{eq:induction} proves the induction step.

Choosing $i=l-1$ and setting $s_0=0$ yields
\begin{align*}
\sum_{\substack{S\subseteq\left[n\right]\colon\\ \left|S\right|=l}}{\prod_{i\in S}{w_i^\delta}}&\le n^{l} \cdot w_1^{l \cdot (\delta+1-\beta)} \cdot d^{l} \tfrac1{l!}.
\end{align*}
Since $w_1^{\delta + 1 - \beta} = \Theta(1)$, we can choose an appropriate constant $c$ such that the statement holds.
\end{proof}
We are now able to show \thmref{sat}. As discussed above, we do this by upper bounding the probability that a bi-cycle appears in $\Phi'$. To this end, we calculate the expected number of bi-cycles in $\Phi'$, observe that it is $\mathrm{poly}(n)^{-1}$, and apply Markov's inequality. This yields that \whp, $\Phi'$ and thus $\Phi$ are satisfiable. 
\begin{proof}[Proof of \thmref{sat}]
We calculate the expected number of bi-cycles in $\Phi'$. First, we fix a set $S\subseteq\left[n\right]$ of $l\ge 2$ variables to appear in a bi-cycle.
Let $X_B$ denote the random variable counting how many times a \emph{specific} bi-cycle $B$ with the variables from $S$ appears in $F$.
Then
\begin{align*}
\mathbb{E}\left[X_B\right]
\le{}&\binom{m}{l+1}(l+1)!\cdot \Pr[u \lor x_1] \Pr[\bar{x_l} \lor v] \cdot\prod_{i=1}^{l-1}{\Pr[\bar{x_i} \lor x_{i+1} ]}.
\end{align*}
The factor $\binom{m}{l+1}(l+1)!$ counts the possible positions of $B$ in~$F$.
By \lemref{shrinked},
\begin{align*}
\mathbb{E}\left[X_B\right] \le m^{l+1}\cdot \left(\tfrac{c_1}{n^2}\right)^{l+1}\cdot\left( w_{|u|} w_{|v|} \prod_{i\in S}w^2_{i} \right)^{{1-\frac12(k-2)(\beta-2)}}\label{eq:expB}
\end{align*}
for some suitable constant $c_1$. Now let $X_S$ denote the random variable counting how many times \emph{any} bi-cycle with the variables from $S$ appears in $F$.
There are $l!$ permutations of the $l$ variables; and $2^l$ combinations of literals on $l$ variables. 
Similarly, literals $u$ and $v$ have $4$ possible sign combinations.
Thus,
\[\mathbb{E}\left[X_S\right]\le m^{l+1}\cdot l!\cdot 2^l\cdot\left(\tfrac{c_1}{n^2}\right)^{l+1}\cdot4\left(\sum_{i\in S}w_i^{{1-\frac12(k-2)(\beta-2)}}\right)^2\prod_{i\in S}{w_i^{2-(k-2)(\beta-2)}}.\]
To estimate the sum, we upper bound $w_i \le w_n$ for all sets up to a certain size $l_0$, which we will determine later. We set $\delta := 2 - (k-2)(\beta-2)$ and define $\alpha(l)$ as
\[\left(\sum_{i\in S}w_i^{\delta/2}\right)^2 \le \alpha(l):=
\begin{cases}
\Oh(l^2), \quad &\text{ if $\delta \leq 0$,} \\
l_0^2\cdot w_n^{\delta}, \quad &\text{ if } \delta > 0 \text{ and } l\le l_0, \\
\Oh(n^2), \quad &\text{ otherwise.}
\end{cases}\]
Now let $X$ denote the random variable counting the number of bi-cycles that appear in $F$.
\vspace{-1mm}
\begin{equation*}
\Ex X \le \sum_{l=2}^{n}{2^{l+2} \cdot m^{l+1}\cdot l!\cdot (\tfrac{c_1}{n^2})^{l+1}  \cdot\alpha(l)\sum_{\substack{S\subseteq\left[n\right] \\ \left|S\right|=l}}{\prod_{i\in S}{w_i^{\delta}}}}.
\end{equation*}
Since $\delta+1 = 2-(k-2)(\beta-2) + 1 < \beta$ by our assumption $\betabound$, we can apply Lemma~\ref{lem:sum}. Using $r:=m/n$, we obtain that the right-hand side is at most
\vspace{-1mm}
\begin{align}
\Ex X &\leq \sum_{l=2}^{n}2^{l+2} \cdot m^{l+1}\cdot l!\cdot (\tfrac{c_1}{n^2})^{l+1}  \cdot\alpha(l) \cdot n^{l}  \cdot c^{l} \tfrac1{l!}  
\leq \tfrac 1n \sum_{l=2}^{n} c_2^l \cdot r^{l} \cdot  \alpha(l), \label{eq:bicycle}
\end{align} 
for some suitable constant $c_2$. Since $r$ is a small enough constant we thus have $c_2 \cdot r < 1$. If $\delta \leq 0$, we are finished, since then
\vspace{-1mm}
\[ \tfrac 1n \sum_{l=2}^{n} c_2^l \cdot r^{l} \cdot  \alpha(l) \leq \tfrac 1n \sum_{l=2}^{n} (c_2\cdot r)^{l} \cdot  l^2 \leq \Oh(\tfrac 1n).\]
Otherwise, if $\delta>0$, we choose
$l_0:=-4\cdot\ln^{-1}(c_2 r)\ln(n),$
which ensures $(r\cdot c_2)^l = \Oh( n^{-4})$ for all $l> l_0$.
For $l=2,\ldots,l_0$, equation~\eqref{eq:bicycle} sums up to at most 
\vspace{-1mm}
\[\tfrac 1n \sum_{l=2}^{l_0} (c_2r)^l \cdot l_0^2 \cdot w_n^{\delta} = \Oh(\log^3(n) \cdot n^{1 - k\frac{\beta-2}{\beta-1}}),\]
where we substituted $w_n = \Theta(n^{\frac{1}{\beta-1}})$ and $\delta = 2 - (k-2)(\beta-2)$. Since $\betabound$, the exponent $1 - k\frac{\beta-2}{\beta-1} < -\eps'$ is negative, and we thus have 
\vspace{-1mm}
\[\Ex X \leq \tfrac 1n \sum_{l=2}^{n} c_2^l r^{l}  \alpha(l) \leq \Oh(\log^3(n) \cdot n^{-\eps'}) + \Oh(\tfrac1{n}),\]
which proves the Theorem by Markov's inequality.
\end{proof}

\section{Discussion of the Results}
In this work, we have shown that with high probability, a power law random $k$-SAT formula is satisfiable, if
$\beta \geq \tfrac{2k-1}{k-1} + \eps$
and the clause-variable ratio is not too large; and that it is unsatisfiable if $\beta \leq \tfrac{2k-1}{k-1} - \eps$, or if the clause-variable ratio is too large. Here, we give a few observations following these results.

First, as explained in Section 1 our results translate directly to the model where clause lengths are power law distributed. This observation might help to explain a phenomenon that arose in \cite{AnsoteguiBL09}: The authors experimentally observed that a random-sat formula with double power law distribution (both variables and clause lengths are drawn from a power law) can be solved extremely fast by MiniSAT. Although the formula was of length $5 \cdot 10^5$, MiniSAT already gave an answer after $4$ seconds! Using our results, we are now able to provide a potential explanation for this phenomenon: Disregarding the double power law distribution, the smallest clause length $k_{\min}$ occurring in their generated formulas is one. Thus, there will be $\Theta(n)$ clauses of length one and by \thmref{unsat} the formula is likely unsatisfiable. 

Second, we observe a sharp threshold in the sense of Friedgut~\cite{Friedgut1999thresholds} (for small constraint densities $r$) for $\beta$ at the point $\tfrac{2k-1}{k-1}$. In contrast, it is unclear whether such a sharp threshold exists (and can be analytically derived) for fixed $\beta$ but variable $r$. Considering however, that decades of research were dedicated to the same question in the uniform case---an arguably simpler model---it is unlikely that we obtain a satisfying answer any time soon; at least for all $k$. As in the uniform model, however, it might be more tractable to get sharp thresholds for $k \to \infty$.

\bibliography{arxiv}

\end{document}